\documentclass[conference]{IEEEtran}
\usepackage{fancyhdr}
\usepackage[comma,numbers,square,sort&compress]{natbib}
\usepackage{mathrsfs}
\usepackage{tikz}
\usepackage{color}
\usepackage{multicol}
\usepackage{amsmath}
\usepackage{amssymb}
\usepackage{amsthm}
\usepackage{graphicx}
\usepackage{times}
\usepackage{graphics,color}
\newtheorem{theorem}{Theorem}	
\newtheorem{lemma}{Lemma}
\newtheorem{proposition}{Proposition}

\newtheorem{definition}{Definition}
\newtheorem{remark}{Remark}
\newcommand{\onetom}{1,\ldots,m}

\newcommand{\ovlam}{\overline{\lambda}}
\newcommand{\unlam}{\underline{\lambda}}
\graphicspath{{figs/}}
\newenvironment{sequation}{\small\begin{equation}}{\end{equation}}
\newenvironment{seqnarray}{\small\begin{eqnarray}}{\end{eqnarray}}
\newenvironment{seqnarray*}{\small\begin{eqnarray*}}{\end{eqnarray*}}

\begin{document}

\title{Event-triggered Stabilization of Coupled Dynamical Systems with Fast Markovian Switching}
%\author{HAN Yujuan, LU Wenlian, CHEN Tianping}
\author{\IEEEauthorblockN{Yujuan Han}
\IEEEauthorblockA{College of Information Engineering\\
Shanghai Maritime University\\
Shanghai, P.~R.~China \\
Email: yjhan@shmtu.edu.cn}
\and
\IEEEauthorblockN{Wenlian Lu}
\IEEEauthorblockA{School of Mathematical Sciences\\
Fudan University\\
Shanghai, P.~R.~China \\
Email: wenlian@fudan.edu.cn}
\and
\IEEEauthorblockN{Tianping Chen\\ }
\IEEEauthorblockA{School of Computer Sciences/Mathematics\\
Fudan University\\
Shanghai, P.~R.~China \\
Email: tchen@fudan.edu.cn}}
% Note: the first argument in the \affiliation command is optional.
% It defines a label for the affiliation which can be used in the \aref
% command. If there is only one affiliation for all authors, then the
% optional argument in the \affiliation command should be suppressed,
% and the \aref command should aslo be removed after each author in
% \author command, in this case the affiliation will not be numbered.

% ��ע�⣺\affiliation�����ĵ�һ�������ǿ�ѡ�ģ�������������\aref�����ı�ǩ��
% ������������ֻ��һ����λ���벻Ҫʹ��\affiliation�����Ŀ�ѡ������ͬʱ������
% ��\author�����е�ÿλ������������Ҳ����ʹ��\aref���ʾ������
% \author{First Author, Second Author, Third Author}
% \affiliation{Chinese Academy of Sciences, Beijing 100190, P.~R.~China\email{ccc@amss.ac.cn}}
% ��ʱ��λǰ���������ֱ��ţ�������������Ҳû�б���

\maketitle
\thispagestyle{fancy}
\fancyhead{}
\lhead{}
\lfoot{}
\cfoot{}
\rfoot{}
\renewcommand{\headrulewidth}{0pt}
\renewcommand{\footrulewidth}{0pt}

\begin{abstract}
In this paper, stability of linearly coupled dynamical systems with feedback pinning is studied. Event-triggered
rules are employed on both diffusion coupling and
feedback pinning to reduce the updating load of the coupled system. Here, both the coupling matrix and the set of pinned-nodes vary with time are induced by a homogeneous Markov chain. %For each node, the diffusion coupling and feedback pinning are set up from the observation of its neighbors' and target's (if pinned) information at the latest event time and the next event time is triggered by some specified criteria.
For each node, the diffusion coupling is set up from the state information of its neighbors' at their latest triggered time and the feedback pinning uses the target's (if pinned) information at the node's latest event time. The next event time is triggered by some specified criteria. Two event-triggering rules are proposed and it is proved that if the system with time-average coupling and pinning gains are stable, the event-triggered strategies can stabilize the system if the switching is sufficiently fast. Moreover, Zeno behaviors are excluded in some cases. Finally, numerical examples are presented to illustrate the theoretical results.
\end{abstract}

%\begin{keywords}
%event-triggered; stability; pinning control; fast switching
%\end{keywords}

% Please remove or comment out the following line if the footnote is not necessary

\section{Introduction}
Control and synchronization of large-scale dynamical systems have received much attention in recent years \cite{Wu}-\cite{Belykh}. In some cases, it is desired to control a complex network  to a homogeneous trajectory of the uncoupled system, and many control strategies are taken into account to stabilize the system. Among them, pinning control is an effective scheme. Due to the interaction of the network, it is not necessary to impose controllers on all nodes. The general idea behind pinning control is to apply some local feedback controllers only to a fraction of nodes while the rest of nodes can be affected through the interactions among nodes \cite{Chen07}-\cite{Han14}.

In most existing works on linearly coupled dynamical systems, each node needs to gather information of its own state and neighborhood's states and update them continuously or in a fixed sampling rate \cite{Wu}-\cite{Han14}. However, as pointed out in \cite{Astrom},
the event-based sampling technique showed better performance than sampling periodically in time for some simple systems. Hereafter, a number of researchers suggested that the event-based control algorithms can be utilized to reduce communication and computation load in networked systems but still maintain control performance \cite{Dimarogonas}-\cite{Lu15pin}. Therefore, the event-based control is particularly suitable for networked systems with limited resources and has attracted wide interests in the scope of distributed control of networked systems.

In some recent papers \cite{Alderisio}-\cite{Lu15pin}, the authors addressed event-triggered algorithms for pinning control of networks. \cite{Alderisio} gave an exponentially decreasing threshold function, hence, the convergence rates of algorithms are predesigned. The event-triggering threshold in \cite{Gao} was prescribed by the distance among states of nodes and target. However, the coupling topology of the network was static. \cite{Lu15pin} employed event-triggered configurations and pinning control terms to realize stability in linearly coupled dynamical systems with Markovian switching in both coupling matrix and pinned node set. However, to realize stability of the switching system, there must exist at least one stable subsystem. If stability cannot be achieved under any subsystem, the switching sequence need to be designed to assure the stability. Especially, in the fast switching theory, by constructing a stable time-average system, the dynamics of switched system can be stable when the switching is fast enough \cite{Frasca}-\cite{Porfiria}.

In the real world, the graph topology of a network may change very quickly by jumps or switches, due to link failures or new creation in a network. So, it is inevitable to study the stability of fast switching systems. Motivated by these works as well as our previous work \cite{Han14}, in this paper, we employ the event-triggered strategy in both coupling configuration and pinning control terms to realize stability in dynamical systems with fast Markovian switching couplings and pinned node set. Hence, all the subsystems among switching can be unstable in this paper. Noticing the significance of the average system in the analysis of stability of fast switching system \cite{Frasca}-\cite{Porfiria}. In this paper, one triggered strategy is given on the average coupling matrix and average pinned node set, the other is given on the time-varying coupling matrix and pinned node set. For each strategy, it is proved that the proposed event-triggered strategy guarantees the stability of the coupled dynamical systems.

This paper is organized as follows. In Sec. \ref{sec-pre}, the underlying problem is formulated. In Sec. \ref{sec-con}, we propose the event-triggering schemes of diffusion configuration and pinning terms to pin the coupled systems to a homogenous pre-assigned trajectory of the uncoupled node system. Numerical simulations are given in Sec. \ref{sec-sim} to verify the theoretical results. Finally, this paper is concluded in Sec. \ref{sec-conclu}.

{\em Notations}: For a matrix $A$, denote $A_{ij}$ the elements of $A$ on the $i$-th row and $j$-th column. $A^{s} = (A+A^{\top})/2$ denotes the symmetry part of a square matrix $A$. For a vector $x$, denote by $x>0$ that every element of $x$ is positive.
%Denote by $A>0(\ge 0)$ that $A$ is positive(semi-positive) definite and so is negative(semi-negative) with $<0$ and $\le 0$.
$I_m$ denotes the identity matrix with dimension $m$. For a matrix $A$, denote by $\bar{\lambda}(A)$ and $\underline{\lambda}(A)$ the largest and smallest eigenvalues in module. For a symmetric matrix $B$, denote its $i$-th largest eigenvalue by $\lambda_i(B)$. The symbol $\otimes$ represents the Kronecker product. $\| A\|$ denotes the matrix norm of $A$ induced by the vector norm $\|\cdot\|$. For a matrix $A$, $\|A\|_{\infty} = \max_i\sum_{j}|A_{ij}|$. In particular, without special notes, $L_{2}$-vector norm is used in this paper and denote it by $\|\cdot\|$, i.e. $\|x\|= \|x\|_2 =��\sqrt{\sum_{i} |x_i|^2}$.
%$tr(A)$ denotes the trace of a square matrix $A$, i.e., $tr(A) = \sum_{i}A_{ii}$.

\section{Preliminaries}\label{sec-pre}

In this paper, we consider a network of linearly coupled dynamical systems with discontinuous diffusions and feedback pinning terms as follows:
%\begin{eqnarray}\nonumber
%\dot{x}_{i}(t)=f(x_{i}(t),t)-c\sum_{j=1}^{m}L_{ij}(\sigma_t)\Gamma[x_{j}(t_{k}^{i})-x_{i}
%(t_{k}^{i})]\\
%-c\epsilon D_i(\sigma_t)\Gamma[x_{i}(t_{k}^{i})-s(t_k^i)]
%,~ t_{k}^{i}\le t<t_{k+1}^{i} \label{cds}
%\end{eqnarray}
\begin{eqnarray}\nonumber
\dot{x}_{i}(t)=f(x_{i}(t),t)-c\sum_{j=1}^{m}L_{ij}(\sigma_t)\Gamma[x_{j}(t_{k_j(t)}^{j})-x_{i}(t_{k}^{i})]\\
-c\epsilon D_i(\sigma_t)\Gamma[x_{i}(t_{k}^{i})-s(t_k^i)]
,~ t_{k}^{i}\le t<t_{k+1}^{i} \label{cds}
\end{eqnarray}
where $x_{i}(t)\in\mathbb R^{n},i=\onetom$ denotes the state vector of node $i$, the continuous map $f(\cdot,\cdot):\mathbb R^{n}\times \mathbb R^+\to\mathbb R^{n}$ denotes the identical node dynamics. $c$ is the uniform coupling strength at each node. $\sigma_t$ denotes the switching rule. $L_{ij}(\sigma_t)=-1$ if $i$ is linked to $j$ otherwise $L_{ij}(\sigma_t)=0$, and $L_{ii}(\sigma_t)=-\sum_{j=1}^{m}L_{ij}(\sigma_t)$. $\Gamma = [\Gamma_{kl}]_{k,l=1}^n\in\mathbb R^{n\times n}$ is the inner configuration matrix with $\Gamma_{kl}\neq 0$ if two nodes are connected by the $k$-th and $l$-th state component respectively.  $D_i(\sigma_t)= 1$ if node $i$ is pinned at time $t$ by a specific node dynamic trajectory $s(t)$ with $\dot{s} = f(s(t),t)$, $s(0) =s_0$, otherwise $D_i(\sigma_t) = 0$. $\epsilon$ is the pinning strength gain over the coupling strength.

%The time $t_k^i$ is referred to the event-triggered time point for node $i$. At time $t$, each node $i$ collects its neighbors' and target's (if pinned) states with respect to an identical time point $t_{k_i(t)}^i$ with $k_{i}(t)=\max\{k': t^{i}_{k'}\le t\}$. This process goes on through all nodes in a parallel fashion.
$t_{k_j(t)}^j$ with $k_{j}(t)=\max\{k': t^{j}_{k'}\le t\}$ is the latest event time of node $j$ at time $t$. Each node takes the latest information of all its neighbors into account in its diffusion coupling term. Hence, for agent $i$ and time $t_k^i<t\le t_{k+1}^i$, if one of its neighbors, for example, denoted by $j$, is triggered at $t=t_{k'+1}^j$ (let $k��$ be the latest event at node $j$ before $t$), then $j$ transfers its current information to $i$ and $x_j(t_{k'}^j)$ in the diffusion coupling term of node $i$ is replaced by $x_j(t_{k'+1}^j)$. This process goes on through all nodes in a parallel fashion.

In this paper, we suppose the switching rule of the coupling topologies and pinned node sets follows a homogeneous continuous Markov chain \cite{Chilina}, denoted by $\sigma_t$. Suppose the state space of $\sigma_t$ is $\mathbb S =\{1,\cdots,N\}$ and its generator $Q = [q_{ij}]_{N\times N}$ is given by:
\begin{eqnarray*}\label{Q}
  \mathbb{P}\{\sigma_{t+\Delta}=j|\sigma_{t}=i\}
  =\left\{\begin{array}{lr}q_{ij}\Delta +o(\Delta), & i\ne j,\\
  1+q_{ii}\Delta+o(\Delta),& i=j,
  \end{array}
  \right.
\end{eqnarray*}
where $\Delta>0$, $\lim_{\Delta \to 0}(o(\Delta)/\Delta)=0$,
$p_{ij}=-\frac{q_{ij}}{q_{ii}}>0$ is the transition probability from state $i$ to
$j$ if $i\ne j$, while $q_{ii}=-\sum_{j=1,j\ne i}^{N}q_{ij}$.

Let $\Delta_{r}$ for $r=0,1,\cdots$ be the successive sojourn time between jumps of $\sigma_t$. Therefore, the sojourn time in state $j$ is exponentially distributed with parameter
$-q_{jj}$. Clearly, $\mathbb E\left[\Delta_n\right]\le \max_j\frac{1}{-q_{jj}}$. Denote $P = [p_{ij}]_{N\times N}$ the probability transition matrix of the embedded discrete-time Markov process of $\sigma_t$.

Denote $\pi(r) = [\pi_1(r),\cdots, \pi_N(r)]$ the state distribution of the process at the $r$-th switching. Then from the Chapman-Kolmogorov equation \cite{Bremaud}, $\pi(r) =\pi(0) P^r$. If the embedded discrete-time Markov process is ergodic, then from \cite{Billingsley}, $P$ is a primitive  and therematrix exists a state distribution $\bar{\pi} = [\bar{\pi}_1,\cdots,\bar{\pi}_N]$ satisfying $\bar{\pi} = \bar{\pi} P$, where $\bar{\pi}_i>0$ for all $i$. From \cite{Horn},  there exist positive numbers $M_0$ and $\kappa<1$ such that
\begin{equation}\label{Markov_pi}
|\pi_j(r) - \bar{\pi}_j| \le M_0 \kappa^r.
\end{equation}
Pick $\pi_j = \frac{\bar{\pi}_j/q_{jj}}{\sum_{l=1}^N \bar{\pi}_l /q_{ll}}, j=1,\cdots,N$. Then, $\pi Q = \mathbf{0}$ and $\pi = [\pi_1,\cdots,\pi_N]$ is the invariant distribution of Markov process $\sigma_t$. Denote $\bar{L} = \sum_{i=1}^m\pi_i L(i)$ and $\bar{D} = \sum_{i=1}^m\pi_i D(i)$ the average matrices.

For the node dynamics map $f$, it is said to belong to some class ${\rm QUAD}(G,\alpha\Gamma,\beta)$ for some positive definite matrix $G\in\mathbb R^{n\times n}$, constants $\alpha\in\mathbb R$, $\beta>0$ and $\Gamma\in\mathbb R^{n\times n}$, if
\begin{eqnarray*}
(u-v)^{\top}G\bigg[f(u,t)-f(v,t)-\alpha\Gamma(u-v)\bigg]
\\
\le-\beta(u-v)^{\top}G(u-v)\label{quad}
\end{eqnarray*}
hold for all $u,v\in\mathbb R^{n}$.
%\begin{remark}\label{remark-1}
%It should be highlighted that if $f(\cdot,t)$ is globally Lipschitz with coefficient $L_f$, %then $f\in {\rm QUAD}(G,\alpha\Gamma,\beta)$ with some constants $\alpha, \beta$.
%\end{remark}
Throughout this paper, we make\\
{\bf Assumption 1:}~
$f(\cdot,t)$ satisfies globally Lipschitz condition with coefficient $L_f$, i.e. $\|f(u,t)-f(v,t)\|\le L_f \|u-v\|$ hold for all $u,v\in\mathbb R^{n}$.

Then, by
\begin{small}
\begin{eqnarray*}
&&(u-v)^{\top} G \left[f(u,t)-f(v,t)-\alpha \Gamma(u-v)\right]\\
&\le & \frac{1}{2}(u-v)^{\top} G^2 (u-v) +\frac{1}{2}\left\|f(u,t)-f(v,t)\right\|^2\\
&&- \alpha (u-v)^{\top}G\Gamma (u-v)\\
&\le& \frac{1}{2}\left(\overline{\lambda}(G)+\frac{L_f^2}{\underline{\lambda}(G)}\right)(u-v)^{\top} G (u-v)\\
&& - \frac{\alpha\underline{\lambda}\left(G\Gamma+\Gamma^{\top}G\right)}
{\overline{\lambda}(G)}(u-v)^{\top} G (u-v)
\end{eqnarray*}
\end{small}	
we have $f\in {\rm QUAD}(G,\alpha\Gamma,\beta)$ with $\alpha \in\mathbb R, \beta=\frac{\alpha\underline{\lambda}\left(G\Gamma+\Gamma^{\top}G\right)}
{\overline{\lambda}(G)}-\left(\frac{\overline{\lambda}(G)}{2}+
\frac{L_f^2}{2\underline{\lambda}(G)}\right)$.

In fact, we do not need the Lipschitz condition hold for all $u,v\in\mathbb R^{n}$ but for a region $\Lambda\subset\mathbb R^{n}$ which contains the global attractors of the coupling system.

For any vector $\xi\in\mathbb R^n$ and matrix $B\in \mathbb R^{n\times n}$, we have $\xi^{\top}B \xi = \sum_{i,j}B_{ij}\xi_i\xi_j\le \sum_{i,j} \frac{|B_{ij}|}{2}\left(\xi_i^2+\xi_j^2\right)$. Hence,
\begin{lemma}\label{lemma-1}
For any vector $\xi\in\mathbb R^n$, any matrix $B\in \mathbb R^{n\times n}$, $\xi^{\top}B \xi\le n\|B\|_{\infty}\xi^{\top}\xi$ holds.
\end{lemma}
%Throughout this paper, we will study the stability of $s(t)$ in the coupled system.
\begin{lemma}
Denote by $\{\tau_r\}_{r\in\mathbb N}$ the time sequence that the topology of the network jumps and $\Delta_r = \tau_{r+1}-\tau_r$. Then
\begin{small}
\begin{align}\nonumber
&\left\|\mathbb E\left[\int_{\tau_r}^{\tau_{r+1}}(L(\sigma_t)-\bar{L})dt\Big|x(\tau_r)\right]\right\|_{\infty} \\\label{est-ave-L}
\le &2N\cdot M_0 \max_{j}\frac{1}{-q_{jj}}\max_{i}\|L(i)\|_{\infty} \kappa^r
\end{align}
\end{small}
and
\begin{small}
\begin{align}\nonumber
 &\left\|\mathbb E \left[ \int_{\tau_r}^{\tau_{r+1}} \left(D(\sigma_{t})-\bar{D}\right)dt \Big| x(\tau_r)\right] \right\|_{\infty}\\ \label{est-ave-D}
 \le & 2N\cdot M_0 \max_{j}\frac{1}{-q_{jj}}\max_{j}\|D(j)\|_{\infty} \kappa^r.
\end{align}
\end{small}
\end{lemma}
\begin{proof}
By (\ref{Markov_pi}) and the definition of $\bar{L}$, we have
\begin{small}
\begin{align*}
&\left\|\mathbb E\left[\int_{\tau_r}^{\tau_{r+1}}(L(\sigma_t)-\bar{L})dt\Big|x(\tau_r)\right]\right\|_{\infty} \\
=&\left\|
\sum_{j=1}^m \pi_{j}(r)L(j) \frac{1}{-q_{jj}} -\bar{L} \mathbb E(\Delta_r)\right\|_{\infty}\\
=& \left\| \sum_{j=1}^m \pi_j(r)L(j) \frac{1}{-q_{jj}}-\sum_{j=1}^m \bar{\pi}_j L(j) \frac{1}{-q_{jj}}\frac{\sum_{j=1}^m \pi_j(r)\frac{1}{-q_{jj}}}{\sum_{j=1}^m \bar{\pi}_j\frac{1}{-q_{jj}}}\right\|_{\infty}\\
\le & 2N\cdot M_0 \max_{j}\frac{1}{-q_{jj}}\max_{j}\|L(j)\|_{\infty}\kappa^r
\end{align*}
\end{small}
Similarly, the second equality (\ref{est-ave-D}) can be derived.
\end{proof}

\begin{definition}
The coupled system is said to be stable at $s(t)$ in mean
square sense, if
$$
\lim_{t\to+\infty}\mathbb{E}\bigg[\|x_{i}(t)-s(t)\|^2\bigg]=0,~i=1,\cdots,m.
$$
\end{definition}

\section{Stability of event triggered algorithms}\label{sec-con}
The stability of the fast switching system heavily relates to an average system, which is a linearly coupled system with coupling matrix $\bar{L}$ and pinned matrix $\bar{D}$. Based on the average matrices, we define the state measurement error for system (\ref{cds}) by
%\begin{small}
%\begin{align}\nonumber
%&e_i(t) =\sum_{j}\bar{L}_{ij}\Gamma\left[x_j(t)-x_j(t_{k_i(t)}^i)-x_i(t)+x_i(t_{k_i(t)}^i)
%\right]\\ \label{error-1}
%&+\epsilon \bar{D}_i\Gamma\left[x_i(t)-s(t)-x_i(t_{k_i(t)}^i)+s(t_{k_i(t)}^i)\right],~i=1,
%\cdots,m.
%\end{align}
%\end{small}
\begin{small}
\begin{align}\nonumber
&e_i(t) =\sum_{j}\bar{L}_{ij}\Gamma\left[x_j(t)-x_j(t_{k_j(t)}^j)-x_i(t)+x_i(t_{k_i(t)}^i)\right]\\ \label{error-1}
&+\epsilon \bar{D}_i\Gamma\left[x_i(t)-s(t)-x_i(t_{k_i(t)}^i)+s(t_{k_i(t)}^i)\right],~i=1,\cdots,m.
\end{align}
\end{small}

In the following, we will demonstrate that if the topologies and pinned node sets switch fast enough and the so-called average system is stable, then the switching system with event-triggered algorithms is stable.
\begin{theorem}\label{thm1}
Suppose $f$ satisfies assumption 1 and  there exists a positive definite matrix $P$ such that $\{P[\alpha I_m -c\bar{L}-c\epsilon \bar{D}]\otimes G\Gamma\}^{s}$ is negative semi-definite. Let $\underline{\lambda}=\underline{\lambda}\left(P\otimes G\right)$ and $\bar{\lambda}=\bar{\lambda}\left(P\otimes G\right)$. Pick a constant $\beta$ satisfying $0<\beta'<\beta$. Set
$t_{k+1}^i$ as the time point defined by the following rule
\begin{seqnarray}\label{event1-1}
\tilde{t}_{k+1}^i &= &\max_{\tau>t_k^i}{\left\{\tau:  \|e_i(\tau)\|  \le \frac{\beta'\underline{\lambda}}{c\overline{\lambda}}\|{x}_i(\tau)-s(\tau) \| \right\}},\\\label{event1-2}
t_{k+1}^i &=& \min\{\tilde{t}_{k+1}^i, t_k^i + T\}.
\end{seqnarray}
If $q_{jj} $ satisfies
\begin{seqnarray} \label{switch-speed-4-rule-1}
-(\beta-\beta')\min_j\frac{1}{-q_{jj}} \cdot e^{\rho_1 \max_{j}\frac{1}{-q_{jj}}}+[2\rho_2(\beta-\beta')+K_1]\\ \nonumber
 \cdot\max_{j}\frac{1}{q^2_{jj}}\cdot M_1(T,\max_{j}\frac{1}{-q_{jj}})
+K_2\cdot \max_{j}\frac{1}{q^2_{jj}}\cdot  e^{\rho_1 \max_{j}\frac{1}{-q_{jj}}}\\ \nonumber
\cdot\left[1+\rho_2 \max_{j}\frac{1}{-q_{jj}}\cdot  M_1(T,\max_{j}\frac{1}{-q_{jj}})\right]
<0
\end{seqnarray}
where $T$ satisfies $h(T)= e^{-\rho_1 T} - \rho_2 T>0$ and
\begin{small}
\begin{align}\nonumber
R(\sigma_t) &= c\bar{L}-cL(\sigma_t)+c\epsilon \bar{D}-c\epsilon D(\sigma_t)],\\\nonumber%\label{def-A-sigma}
A(\sigma_t) &= \{P(c\bar{L}- cL(\sigma_t)-c\epsilon\bar{D}+c\epsilon D(\sigma_t))\otimes G\Gamma\}^{s},\\ \nonumber
m_0 &=c\cdot  \max_{i,j} \left(L_{ii}(j)+\epsilon D_{i}(j)\right),\\ \nonumber
m_1 &= c\cdot\max_{i,j,q}\left\{\left|L_{ij}(q)-\bar{L}_{ij}\right|,
\epsilon\left|D_i(q)-\bar{D}_i\right|\right\},\\ \nonumber
K_1 &= \Big\{m^2\cdot m_0\max_{j}\|P\cdot R(j)\|_{\infty}
+ m_1 \|P\|_{\infty}[3m(m+1) L_f^2\\ \nonumber
&~~~~+(m+m^2+m_0m(m+1)^2)\ovlam\left(\Gamma^{\top}GG\Gamma\right)\\\nonumber%\label{def-K-1}
&~~~~+(m_0+1)(m+m^2)\ovlam\left(G\Gamma\Gamma^{\top}G\right)
\Big\}\frac{1}{\unlam},\\ \nonumber
K_2 &= \Big\{\max_j\ovlam\left(A(j)A^{\top}(j)\right)+ L_f^2 + 2m\cdot m_0 \max_{j}\|PR(j)\|_{\infty}\\ \nonumber %\label{def-K-1}
&~~~~~\cdot\ovlam(\{G\Gamma\}^{s}\Gamma\Gamma^{\top}
\{G\Gamma\}^{s})+ 3m_0m(m+1)^2+m_1 \|P\|_{\infty}\\ \nonumber%\label{def-K-2}
&~~~~\cdot\big[(m+m^2+m_0m(m+1)^2)\ovlam\left(\Gamma^{\top}GG\Gamma\right)\big]\Big\}\frac{1}{\unlam}, \\ \nonumber%\label{def-rho-1}
\rho_1 &= 2\left|\alpha\frac{\overline{\lambda}\left(\left\{G\Gamma\right\}^{s}\right)}
{\underline{\lambda}(G)}-\beta\right|+2m_0\frac{\overline{\lambda}
\left(G\Gamma\Gamma^{\top}G\right)}{\underline{\lambda}(G)},\\ \nonumber %\label{def-rho-2}
\rho_2 &= \frac{\sum_{i}P_{ii}(cL_{ii}+c\epsilon D_i)}{\unlam},\\ \nonumber%\label{def-M-4}
M_1(T,\Delta) &= \max\big\{\frac{1}{h(T)}, [1-\rho_2\Delta \cdot e^{\rho_1\Delta}]^{-1} e^{\rho_1 \Delta}\big\},
\end{align}
\end{small}
then under the updating rule (\ref{event1-1}) and (\ref{event1-2}), system (\ref{cds}) is stable at the homogeneous trajectory $s(t)$ in the mean square sense.
\end{theorem}
{\em Brief~proof.}
Let $\hat{x}_i(t) \triangleq x_i(t) - s(t)$, then the dynamics of $\hat{x}_i(t)$ in $[t^{i}_{k}, t^{i}_{k+1})$ is
\begin{small}
\begin{eqnarray}\nonumber
\dot{\hat{x}}_{i}(t)
= f(\hat{x}_{i}(t)+s(t),t)-f(s(t),t)-c\sum_{j=1}^{m}L_{ij}(\sigma_t)\Gamma\\   \nonumber \label{cds2}
\cdot\left[\hat{x}_{j}(t)-\hat{x}_{i}(t)\right]
-c\epsilon D_i(\sigma_t)\Gamma\hat{x}_{i}(t)+\hat{e}_i(t)
\end{eqnarray}
\end{small}
with
\begin{small}
\begin{eqnarray} \nonumber
\hat{e}_i(t) = \sum_{j}L_{ij}(\sigma_t)\Gamma\left[x_j(t)-x_j(t_{k_j(t)}^j)-x_i(t)+x_i(t_{k_i(t)}^i)\right]\\ \label{error-2}
+\epsilon {D}_i(\sigma_t)\Gamma\left[x_i(t)-s(t)-x_i(t_{k_i(t)}^i)+s(t_{k_i(t)}^i)\right]
\end{eqnarray}
\end{small}
Let
\begin{align*}
\hat{x}(t) = [\hat{x}^{\top}_1(t),\cdots,\hat{x}^{\top}_m(t)]^{\top},~~\hat{e}(t) = [\hat{e}^{\top}_1(t),\cdots,\hat{e}^{\top}_m(t)]^{\top}
\end{align*}
 and $V(t) = \frac{1}{2} \hat{x}^{\top}(t)(P\otimes G)\hat{x}(t)$. By $f\in {\rm QUAD}(G,\alpha\Gamma,\beta)$, $ \|e_i(\tau)\|  \le \frac{\beta'\underline{\lambda}}{c\overline{\lambda}}\|{x}_i(\tau)-s(\tau) \|$, Cauchy-
Schwarz inequality, inequalities (\ref{est-ave-L}), (\ref{est-ave-D}) and Lemma \ref{lemma-1}, we have
\begin{small}
\begin{align}\label{est-V}
& \mathbb E\left[V(\tau_{r+1}) |\hat{x}(\tau_r)\right]- V(\tau_r)\\ \nonumber
\le &-2(\beta-\beta')\mathbb E\left[\Delta_r\cdot e^{-\rho_1\Delta_r}\right] V(\tau_{r})\\ \nonumber
&+ 2(\beta-\beta')\rho_2\cdot \mathbb E\left[ \Delta_r^2 \cdot M_1(T,\Delta_r) \right]V(\tau_{r})\\ \nonumber
&+[M_2+ M_3\cdot \mathbb E\left[M_1(T,\Delta_r)\right]]
\cdot C_0\cdot \kappa^r  V(\tau_{r})\\ \nonumber
&+ K_1 \cdot \mathbb E\left[\Delta_r^2\cdot  M_1(T,\Delta_r)\right] V(\tau_r)  \\ \nonumber
&+K_2\cdot \mathbb E\big[\Delta_r^2\cdot e^{\rho_1 \Delta_r}+\rho_2\Delta_r^3 \cdot e^{\rho_1 \Delta_r} M_1(T,\Delta_r)\big] V(\tau_r)
\end{align}
\end{small}
with
\begin{small}
\begin{align*}
C_0 &= 2N\cdot M_0 \max_{j}\frac{1}{-q_{jj}}\max_{i}\left\{\|L(i)\|_{\infty}, \|D(i)\|_{\infty}\right\},\\
M_2 &= \Big\{2mn(c+c\epsilon )\| P \|_{\infty} \|G\Gamma\|_{\infty} C_0 \\
&~~~~+ cm(m+\epsilon)\left(\ovlam\left(G\Gamma\Gamma^{\top}G\right)+2\right)\Big\}\frac{1}{\unlam},\\
M_3 &= 2cm(m+\epsilon)\frac{1}{\unlam}.
\end{align*}
\end{small}
 The property of Markov process gives
  \begin{small}
  \begin{align}\nonumber
    &\mathbb E\left[\Delta_r\right] = \sum_{j=1}^m \pi_{j}(r) \frac{1}{-q_{jj}} \le \max_j\frac{1}{-q_{jj}},\\\label{exp-sojourn-time-1}
    &\mathbb E\left[\Delta_r\cdot e^{-\rho_1\Delta_r}
  \right]\ge \min_j\frac{1}{-q_{jj}}e^{-\rho_1\max_j
  \frac{1}{-q_{jj}}},\\ \label{exp-sojourn-time-2}
  &\mathbb E\left[\Delta_r^2\cdot M_1(T,\Delta_r)\right]
  \le \max_j \frac{1}{q^2_{jj}}M_1(T,\max_j\frac{1}{-
  q_{jj}}),
  \end{align}
  \end{small}
and
\begin{small}
\begin{align}\label{exp-sojourn-time-3}
&\mathbb E\left[\Delta_r^2\cdot e^{\rho_1 \Delta_r}+\rho_2\Delta_r^3 \cdot e^{\rho_1 \Delta_r} M_1(T,\Delta_r)\right]\\ \nonumber
\le & \max_j\frac{1}{q^2_{jj}}\cdot e^{\rho_1 \max_j\frac{1}{-q_{jj}}}\left[1+\rho_2 \max_{j}\frac{1}{-q_{jj}}\cdot  M_1(T,\max_{j}\frac{1}{-q_{jj}})\right].
\end{align}
\end{small}
By $\kappa<1$, we have that for any $0<\delta<1$, there exists $N_1(\delta)$, such that for any $r>N_1(\delta)$,
\begin{small}
\begin{equation}\label{est-switch-times}
\left[M_2 + M_3\cdot M_1\left(T,\max_j\frac{1}{-q_{jj}}\right)\right]\cdot \kappa^r \le {\delta}.
\end{equation}
\end{small}
Apply assumption (\ref{switch-speed-4-rule-1}) and estimations (\ref{exp-sojourn-time-1})-(\ref{est-switch-times}) to (\ref{est-V}), we get that there exists some $\delta_0>0$ such that for all $0<\delta<\delta_0$,
\begin{small}
\begin{align}\label{con-est-V-5}
&\mathbb E\left[V(\tau_{r+1}) |\hat{x}(\tau_r)\right]- V(\tau_r) \\ \nonumber
\le& -(\beta-\beta')\min_j\frac{1}{-q_{jj}}\cdot e^{-\rho_1\max_j \frac{1}{-q_{jj}}}V(\tau_r).
\end{align}
\end{small}
Take expectation on both sides of (\ref{con-est-V-5}),
\begin{small}
\begin{align*}
&\mathbb E\left[V(\tau_{r+1})\right]\\
\le& \left[1-(\beta-\beta')\min_j\frac{1}{-q_{jj}}\cdot e^{-\rho_1\max_j \frac{1}{-q_{jj}}}\right]\mathbb E\left[V(\tau_r)\right].
\end{align*}
\end{small}
 which implies for $r>N_1(\delta)$, $\mathbb E\left[V(\tau_{r})\right]$ decreases as $r$ increases. By the assumption on $f$, it can be calculated that
\begin{sequation}\label{con-est-V-6}
V(t) \le e^{\rho_1 \Delta_r}V(\tau_{r}) + \rho_2 e^{\rho_1 \Delta_r}M_1(T,\Delta_r)V(\tau_r)
\end{sequation}
Take expectation on both sides of (\ref{con-est-V-6}), we have
\begin{small}
\begin{align*}
\mathbb E\left[V(t)\right]
\le \mathbb E \left[V(\tau_{r})\right] e^{\rho_1\max_j\frac{1}{-q_{jj}}}
 \left[1+ \rho_2 M_1(T, \max_{j}\frac{1}{-q_{jj}})\right].
\end{align*}
\end{small}
Combing with $\lim_{r\in\mathbb N, r\to+\infty}\mathbb E\left[V(\tau_{r})\right]=0$, we have $\mathbb E[V(t)]$ converges to 0 as time goes to infinity. From the non-negative property of $V(t)$, we can conclude that the system is stable at $s(t)$ in the mean square sense. This completes the proof.
\qed
\begin{remark}
Noting the sojourn time in the subsystem with coupling matrix $L(j)$ and pinned matrix $D(j)$ is exponentially distributed with parameter $-q_{jj}$, $j\in\mathbb S$. Therefore, for any $j\in\mathbb S$, the larger $-q_{jj}$ is, the shorter the sojourn time in the subsystem with matrices $L(j), D(j)$ is. %In all, the value $\max_j\{-q_{jj}\}$ reflects the speed of switching among all subsystems.
By Theorem \ref{thm1}, it can be seen that to satisfy condition (\ref{switch-speed-4-rule-1}), the parameters $-q_{jj}, j=1,\cdots,N$ should be sufficiently large, which induces fast switching among subsystems.
\end{remark}
\begin{remark}

By Theorem \ref{thm1}, it is clear that the stability of the fast switching system closely relates to the stability of the average system. If the invariant distribution of the Markov process, denoted by $\pi$, satisfies $\pi>0$, then the network topology of the average system is the union of all possible graph topologies in the switching system. In \cite{Chen07}, it was proved that if the network topology is strongly connected, the linearly coupled system can be stabilized by a single pinned controller. Here it can be derived that if the union of all possible graph topologies is strongly connected and the invariant distribution of the Markov process satisfies $\pi>0$, then under sufficiently fast switching, system (\ref{cds}) with event-triggered diffusions and pinned terms is stable.
\end{remark}
\begin{remark}
To assure the decreasing of $\mathbb E \left[V(t)\right]$, an upper bound of the inter-event interval, denoted by $T$, is imposed to all nodes. Practically, if the time calculated from rule (\ref{event1-1}) is $T$ time longer than current triggered time, an externally triggering will be applied to the node to guarantee $t_{k+1}^i-t_k^i\le T$ hold for all $k$ and $i$.
\end{remark}

In the above theorem, the updating rule (\ref{event1-1}) and (\ref{event1-2}) is defined by $e_i(t), i=1,\cdots,m$ in (\ref{error-1}). Next, we propose another updating rule based on $\hat{e}_i(t), i=1,\cdots,m$ in (\ref{error-2}) to stabilize the system.
\begin{theorem}\label{thm2}
Suppose all assumptions in Theorem \ref{thm1} hold. For $0<\beta'<\beta$, set $t_{k+1}^i$ as the time point defined by the following rule
\begin{seqnarray}\label{event2}
t_{k+1}^i =\max_{\tau>t_k^i}{\left\{\tau:  \|\hat{e}_i(\tau)\|  \le \frac{\beta'\underline{\lambda}}{c\overline{\lambda}}\|{x}_i(\tau)-s(\tau) \| \right\}}
\end{seqnarray}
where $\hat{e}_i(t)$ is defined in (\ref{error-2}).
If $q_{jj}$ satisfies
\begin{align}\nonumber%\label{switch-speed-4-rule-2}
&-(\beta-\beta')\min_j\frac{1}{-q_{jj}} \cdot e^{\rho_1' \max_{j}\frac{1}{-q_{jj}}}\\ \nonumber
&+K_1'\max_{j}\frac{1}{-q_{jj}}\cdot  e^{\rho_1' \max_{j}\frac{1}{-q_{jj}}}
<0
\end{align}
where
\begin{small}
\begin{align*}%\label{def-K-1'}
K_1' =& \Big[L_f^2+(c+1)\max_{j}\ovlam\left(A(j)A^{\top}(j)\right)\\\nonumber
&+ mn \max_j\left\|A(j)(\left(cL(j)+c\epsilon D(j)\right)\right\|_{\infty}\Big]\frac{1}{\unlam}
+\frac{\beta'^2\unlam}{c\ovlam^2},\\
%\label{def-rho-1'}
\rho_1' =& 2(\beta+\beta')+2\max_{i,j}\Big|\lambda_j\big(\{P(\alpha I_m + cL(i)\\ \nonumber
&+c\epsilon D(i))\otimes G\Gamma\}^{s}\big)\Big|\frac{1}{\unlam}
\end{align*}
\end{small}
then under the updating rule (\ref{event2}), system (\ref{cds}) is stable at the homogeneous trajectory $s(t)$ in the mean square sense.
\end{theorem}
The proof of this Theorem is similar to that of Theorem \ref{thm1}. Here it is omitted for saving space.
\begin{remark}
Comparing the updating rule (\ref{event1-1}) and (\ref{event1-2}) and rule (\ref{event2}), it can be seen that the upper bound of triggered time intervals of each node, denoted by $T$ in (\ref{event1-2}), is no longer required in the updating rule (\ref{event2}) defined by $\hat{e}_i(t)$. However, as a trade off, the triggering events under rule (\ref{event2}) happen more frequently than rule (\ref{event1-1}) and (\ref{event1-2}) if the switching among topologies is sufficiently fast. Since $\hat{e}_i(t)$ is decided by the varying coupling matrix $L(\sigma_t)$ and pinned matrix $D(\sigma_t)$, which jumps at the time point when the topology switches. Therefore, for systems with updating rule (\ref{event2}), each node has to update its state information at the switching time points of the topologies. If the switching among topologies is sufficiently fast, the triggered times of all nodes will substantially increase.
\end{remark}

Similar to work \cite{Dimarogonas}, under the above two updating rules (\ref{event1-1}), (\ref{event1-2}) and (\ref{event2}), there exists at least one node with its next inter-event interval being strictly positive. Under some special hypothesis, the Zeno behavior \cite{Johansson} can be excluded for all nodes.
\begin{proposition}\label{proposition-1}
Suppose all hypotheses of Theorem 1 hold.
\begin{enumerate}
  \item Under either of the updating rule (\ref{event1-1}) and (\ref{event1-2}) and rule
      (\ref{event2}), if the system does not converge, there exists at least one node $i$ such that its next inter-event interval is strictly positive.
  \item Suppose there exists some $\eta$ (possibly negative) such that
\begin{equation}\nonumber
  (u-v)^{\top}(f(u)-f(v))\ge \eta (u-v)^{\top}(u-v)
\end{equation}
  holds for all $u, v \in \mathbb R^n$. If there exists a constant $b>0$ such that
  $\|\hat{x}_i(t)\|^2\ge b V(t)$ holds for some $i\in\{1, 2,\cdots, m\}$, then the next inter-event interval of node $i$ is strictly positive and lower bounded
  by a common constant.
\end{enumerate}
\end{proposition}
The proof of this proposition is omitted here for saving space.
\section{simulation} \label{sec-sim}
In this section, we present numerical examples to illustrate the theoretical results. In these examples, we consider three-dimensional neural network as the uncoupled node dynamics \cite{Zou}:
\begin{equation}\label{neur}\nonumber
\frac{dx}{dt}=-Dx+Hg(x)
\end{equation}
with $x=(x_1,x_2,x_3)^{\top}\in\mathbb{R}^3$,
\begin{eqnarray*}
H=\left[\begin{array}{ccc}
1.2500 & -3.200 & -3.200 \\
-3.200 & 1.100 & -4.400 \\
-3.200 & 4.400 & 1.000
\end{array}\right]
\end{eqnarray*}
$D=I_3$, $g(x)=(g(x_1),g(x_2),g(x_3))^{\top}$, and $g(s)=(|s+1|-|s-1|)/2$. This system has a double-scrolling chaotic attractor with initial value $x_1(0)=x_2(0)=x_3(0)=0.1000$ \cite{Zou}. Noting that $f(x) = -Dx + H g(x)$ has $9$
Jacobin matrices and $\chi = 4.6769$ is the upper bound of the matrices of $f$. Hence, we estimate $\beta' = \alpha-1/2-\chi^2/2$ and $G = I_3$.
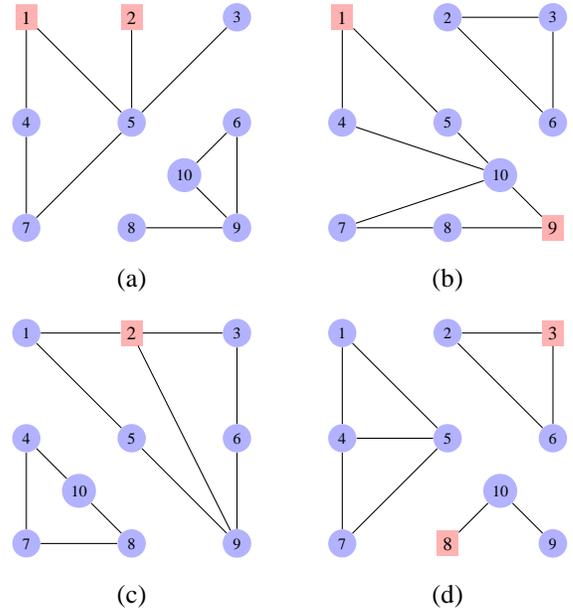
\begin{figure}
\centering
\begin{tikzpicture}[scale=1.4]
\draw
(-1 , 1) node(1)[rectangle, scale=0.7, fill=red!30!white]{1}
(0 , 1) node(2)[rectangle, scale=0.7,fill=red!30!white]{2}
(1 , 1) node(3)[circle, scale=0.6,fill=blue!30!white]{3}
(-1 , 0) node(4)[circle, scale=0.6,fill=blue!30!white]{4}
(0 , 0) node(5)[circle, scale=0.6,fill=blue!30!white]{5}
(1 , 0) node(6)[circle, scale=0.6,fill=blue!30!white]{6}
(-1 , -1) node(7)[circle, scale=0.6,fill=blue!30!white]{7}
(0 , -1) node(8)[circle, scale=0.6,fill=blue!30!white]{8}
(1, -1) node(9)[circle, scale=0.6,fill=blue!30!white]{9}
(0.5, -0.5 ) node(10)[circle, scale=0.6,fill=blue!30!white]{10}

(0,-1.5) node(tt)[]
{(a)};
\draw 
(1) -- (4)
(1) -- (5)
(5) -- (2)
(5) -- (3)
(4) -- (7)
(5) -- (7)
(6) -- (10)
(6) -- (9)
(9) -- (10)
(8) -- (9);

\draw
(2 , 1) node(11)[rectangle, scale=0.7,fill=red!30!white]{1}
(3 , 1) node(12)[circle, scale=0.6,fill=blue!30!white]{2}
(4 , 1) node(13)[circle, scale=0.6,fill=blue!30!white]{3}
(2 , 0) node(14)[circle, scale=0.6,fill=blue!30!white]{4}
(3 , 0) node(15)[circle, scale=0.6,fill=blue!30!white]{5}
(4 , 0) node(16)[circle, scale=0.6,fill=blue!30!white]{6}
(2 , -1) node(17)[circle, scale=0.6,fill=blue!30!white]{7}
(3 , -1) node(18)[circle, scale=0.6,fill=blue!30!white]{8}
(4, -1) node(19)[rectangle, scale=0.7,fill=red!30!white]{9}
(3.5,-0.5)node(20)[circle, scale=0.6,fill=blue!30!white]{10}
(3,-1.5) node(tt)[]
{(b)};
\draw 
(11) -- (14)
(11) -- (15)
(15) -- (20)
(14) -- (20)
(12) -- (13)
(13) -- (16)
(12) -- (16)
(17) -- (20)
(19) -- (20)
(18) -- (19)
(17) -- (18);

\draw
(-1 , -2) node(21)[circle, scale=0.6,fill=blue!30!white]{1}
(0 , -2) node(22)[rectangle, scale=0.7,fill=red!30!white]{2}
(1, -2) node(23)[circle, scale=0.6,fill=blue!30!white]{3}
(-1 , -3) node(24)[circle, scale=0.6,fill=blue!30!white]{4}
(0 , -3) node(25)[circle, scale=0.6,fill=blue!30!white]{5}
(1 , -3) node(26)[circle, scale=0.6,fill=blue!30!white]{6}
(-1 , -4) node(27)[circle, scale=0.6,fill=blue!30!white]{7}
(0 , -4) node(28)[circle, scale=0.6,fill=blue!30!white]{8}
(1, -4) node(29)[circle, scale=0.6,fill=blue!30!white]{9}
(-0.5,-3.5)node(30)[circle, scale=0.6,fill=blue!30!white]{10}
(0,-4.5) node(tt)[]
{(c)};
\draw 
(21) -- (22)
(21) -- (25)
(22) -- (23)
(23) -- (26)
(26) -- (29)
(22) -- (29)
(25) -- (29)
(24) -- (30)
(28) -- (30)
(27) -- (28)
(24) -- (27);

\draw
(2, -2) node(31)[circle, scale=0.6,fill=blue!30!white]{1}
(3 , -2) node(32)[circle, scale=0.6,fill=blue!30!white]{2}
(4, -2) node(33)[rectangle, scale=0.7,fill=red!30!white]{3}
(2 , -3) node(34)[circle, scale=0.6,fill=blue!30!white]{4}
(3 , -3) node(35)[circle, scale=0.6,fill=blue!30!white]{5}
(4 , -3) node(36)[circle, scale=0.6,fill=blue!30!white]{6}
(2 , -4) node(37)[circle, scale=0.6,fill=blue!30!white]{7}
(3 , -4) node(38)[rectangle, scale=0.7,fill=red!30!white]{8}
(4, -4) node(39)[circle, scale=0.6,fill=blue!30!white]{9}
(3.5,-3.5)node(40)[circle, scale=0.6,fill=blue!30!white]{10}
(3,-4.5) node(tt)[]
{(d)};
\draw 
(31) -- (34)
(31) -- (35)
(33) -- (36)
(34) -- (35)
(34) -- (37)
(35) -- (37)
(32) -- (33)
(32) -- (36)
(38) -- (40)
(39) -- (40);
\end{tikzpicture}
\caption{The topologies of the graph of the coupled system and the pinned sets. Pinned sets (a) $\{1,2\}$, (b) $\{1,9\}$, (c) $\{2\}$, (d) $\{3,8\}$.} \label{topology}
\end{figure}

%\begin{figure}
%\begin{center}
%\includegraphics[width=0.5\textwidth]{./figs/rho.eps}
%\includegraphics[width=1.6\textwidth]{topo.eps}
%\caption{The topologies of the graph of the coupled system and the pinned sets. Pinned sets (a) $\{1,2\}$, (b) $\{1,9\}$, (c) $\{2\}$, (d) $\{3,8\}$.} \label{topology}
%\end{center}
%\end{figure}
The possible coupling graph topologies and pinned node sets are shown in Fig. \ref{topology}, here $m =10$. Noting that subsystems with every possible network topology and pinned nodes cannot be stabilized to the target trajectory. Theorems \ref{thm1},\ref{thm2} indicate that if the average system can be stabilized and the switching among subsystems is sufficiently fast, then under the event-triggered strategies, the switching system can be stabilized. Via the following numerical simulations, it can be seen that the switching system can be stabilized.

Suppose the generator of the Markov chain $\sigma_t$ is
$
Q = \left(
         \begin{array}{cccc}
           -100 & 35 & 0 & 65 \\
           40 & -100 & 60 & 0 \\
           0 & 50 & -100 & 50 \\
           30 & 70 & 0 & -100 \\
         \end{array}
       \right),
$
then the sojourn time in each topology follows the exponential distribution with parameter $p = 0.01$.% Hence, the invariant distribution of $\sigma_t$ is $\bar{\pi} = [0.2096,0.3443,0.2066,0.2395]$.

In the following examples, suppose the inner coupling matrix $\Gamma=I_3$, $c = 5$, $\beta' = 1$, $\epsilon = 3$. The ordinary differential equation (\ref{cds}) is numerically solved by the Euler method with a time step 0.001(seconds) and the time duration of the numerical simulations is $[0,10]$ (seconds).
\begin{figure}[!t]
\begin{center}
\includegraphics[width=0.5\textwidth]{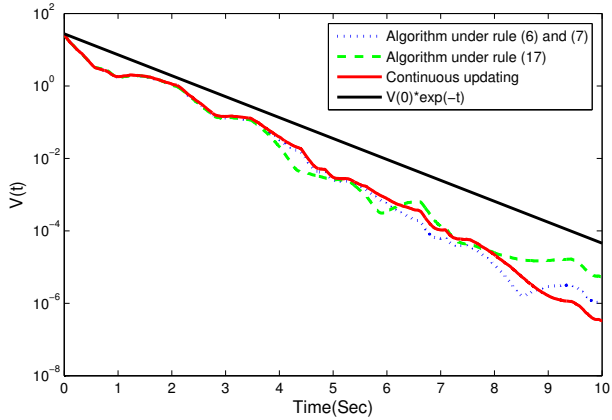}
\caption{The dynamics of Lyapunov function V(t) for systems with event triggering algorithms under rule (\ref{event1-1}) and (\ref{event1-2}) and rule (\ref{event2}) and system with continuous updating.} \label{variation_V}
\end{center}
\end{figure}
Firstly, we employ rule (\ref{event1-1},\ref{event1-2}). Here, $T=0.02$. Fig. \ref{variation_V} shows the dynamics of $V(t)$, which implies that the coupled system (\ref{cds}) is stable. Secondly, rule (\ref{event2}) is considered. The dynamics of $V(t)$ is also given in Fig. \ref{variation_V}, which also implies the stability of system (\ref{cds}). One can see that the coupled system (\ref{cds}) is asymptotically stable at certain  homogeneous trajectory.
\begin{figure}[!t]
\begin{center}
\includegraphics[width=0.5\textwidth]{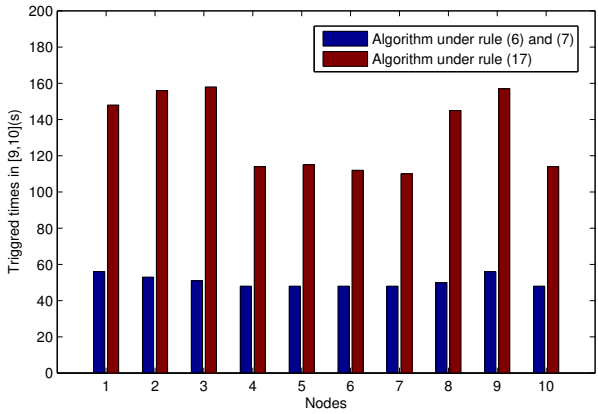}
\caption{Histogram of triggering times of each node in $[9,10]$s under updating rule (\ref{event1-1}) and (\ref{event1-2}) and rule (\ref{event2}).} \label{event-times}
\end{center}
\end{figure}
Furthermore, it can be seen from Fig. \ref{event-times} that the events of updating the diffusion and pinning terms under rule (\ref{event2}) happen more than rule (\ref{event1-1}) and (\ref{event1-2}), as a consequence of fast switching among network topologies.

\section{Conclusions}\label{sec-conclu}
In this paper, event-triggered configurations and feedback pinning are employed to realize stability in linearly coupled dynamical systems with fast Markovian switching, which reduces communication and computation loads. Once an event for a node is triggered, the diffusion coupling term and feedback
control (if pinned) of this node will be updated. Event triggering criteria are derived for each node that can be computed in a parallel way. Two event-triggered rules are proposed and proved to perform well and can exclude Zeno behaviors in some cases. Simulations are given to verify these theoretical results.

\section*{Acknowledgment}
This work was jointly supported by the National Natural Sciences Foundation of China under Grant Nos. 61273211 and 61673119 and the Program for New Century Excellent Talents in University (NCET-13-0139).

\end{document}